\documentclass{article}

\usepackage{oldlfont,amssymb,epsf,stmaryrd}

\newtheorem{definition}{Definition}[section]

\newtheorem{proposition}{Proposition}[section]
\newenvironment{example}{\noindent{\em Example:}}{\bigskip}
\newenvironment{remark}{\noindent{\em Remark:}}{\bigskip}
\newenvironment{proof}{\noindent{\em Proof:}}{$\Box$ \bigskip}

\begin{document}

\title{A computational definition of the notion of vectorial space}
\author{Pablo Arrighi\thanks{
Institut Gaspard Monge,
5 Bd Descartes, Champs-sur-Marne,
77574 Marne-la-Vall\'ee Cedex 2, France,
{\tt arrighi@univ-mlv.fr.}}
\and
Gilles Dowek\thanks{
\'Ecole polytechnique and INRIA,
LIX, \'Ecole polytechnique,
91128 Palaiseau Cedex, France,
{\tt Gilles.Dowek@polytechnique.fr.}}}
\date{\vspace{-.5in}}

\maketitle

\begin{abstract}
We usually define an algebraic structure
 by a set, some operations defined on this
set and some propositions that the algebraic structure must validate.
In some cases, we can replace these propositions
by an algorithm on terms constructed upon these operations
that the algebraic structure must validate. We show in this note that this is 
the case for the notions of vectorial space and bilinear operation.
\end{abstract}


\thispagestyle{empty}

An algorithm defined by a confluent and terminating rewrite system $R$
on terms of a language 
${\cal L}$ is said to be {\em valid} in a structure ${\cal M}$ on the
language ${\cal L}$ if for each rule $l \longrightarrow r$ and
assignment $\phi$, we have $\llbracket l
\rrbracket_{\phi} = \llbracket 
r \rrbracket_{\phi}$.  Thus, algorithms and theories play the same
role with respect to the notion of model: like a theory, an algorithm
may or may not be valid in a model. This notion of validity of
an algorithm, like the notion of validity of a theory, can be used in
two ways: to study the algorithms or to define algebraic
structures as models of some algorithm.

When a class of algebraic structures --- such as the class of groups
or that of rings --- can be defined as the class of models of some
equational theory 
$T$ and this equational theory can be transformed into a
rewrite system $R$, we have the following equivalence
\begin{itemize}
\item $A$ is a member of the class (i.e. is a group, a ring, ...), 
\item $A$ is a model of the theory $T$, 
\item $A$ is a model of the algorithm $R$.
\end{itemize}
In this case, we say that the class of algebraic structures has a
{\em computational definition}. 

The goal of this note is to show that the class of vectorial spaces
has such a computational definition, {\em i.e.} that the axioms of vectorial
spaces can be oriented as a rewrite system. Moreover, the algorithm
obtained this way is a well-known algorithm in linear algebra: it is
an algorithm transforming any term expressing a vector into a linear
combination of the unknowns. This algorithm is also central to the
operational semantic of our functional programming language for quantum
computing {\tt Lineal} \cite{Helsinki}, because in such 
languages a program and its input value form a term expressing a
vector whose value, the output, is a linear combination of the base vectors. 
More generally, several algorithms used in linear algebra, such as matrix
multiplication algorithms, transform a term expressing a vector with
various constructs into a linear combination of base vectors.
This algorithm is valid in all
vectorial spaces and 
we show that it moreover completely defines the notion 
of vectorial space. 

The main difficulty to orient the theory of vectorial spaces is that
this theory has a sort for vectors and a sort for scalars and
that the scalars must form a field. The theory of fields is already
difficult to orient, because division is a partial 
operation. However, there are many fields, for instance the field 
${\mathbb Q}$ of rational numbers, whose addition and multiplication
can be presented by a terminating and ground confluent rewrite
system. Thus, we shall not consider an arbitrary vectorial space over
an arbitrary field. Instead, we consider a given field ${\cal K}$ defined
by a terminating and ground confluent rewrite system $S$ and focus on
${\cal K}$-vectorial spaces. Our rewrite system for vectors
will thus be parametrized by a rewrite system for scalars and
we will have to provide proofs of confluence and termination using 
minimal requirements on the scalar rewrite system. This leads to
a new method to prove the confluence of a rewrite system built as 
the union of two systems.

Moreover, this computational definition of the notion of vectorial
space can be extended to define other algebraic notions such as
bilinear operations.

\section{Rewrite systems}

\begin{definition} [Rewriting]
Let ${\cal L}$ be a first-order language 
and $R$ be a rewrite system on 
${\cal L}$. We say that a term $t$ 
{\em $R$-rewrites} in one step to a term $u$
if and only if  there is an occurrence $\alpha$ in the 
term $t$, a rewrite rule $l \longrightarrow r$ in $R$,
and a substitution $\sigma$ such that
$t_{|\alpha} = \sigma l$ and 
$u = t[\sigma  r]_{\alpha}$.
\end{definition}

\begin{definition} [Associative-Commutative Rewriting]
Let ${\cal L}$ be a first-order language containing binary 
function symbols $f_{1}, ..., f_{n}$ and $R$ be a rewrite system on 
${\cal L}$. We say that a term $t$ 
{\em $R/AC(f_{1}, ..., f_{n})$-rewrites} in one step to a term $u$
if and only if  there is a term $t'$, an occurrence $\alpha$ in the 
term $t'$, a rewrite rule $l \longrightarrow r$ in $R$,
and a substitution $\sigma$ such that
$t' =_{AC} t$, $t'_{|\alpha} = \sigma l$ and 
$u =_{AC} t'[\sigma  r]_{\alpha}$.
\end{definition}

\begin{remark}
This notion must be distinguished from that of {\em
R,AC-rewriting} \cite{PetersonStickel} where a term $t$ rewrites to a term $u$ 
only when it has a subterm AC-equivalent to an instance of the 
left hand side of a rewrite rule. For instance with the rule 
$x + x \longrightarrow 2 . x$ the term 
$t + (u + t)$ $R/AC$-rewrites 
to $2.t + u$ but is $R,AC$-normal. 
\end{remark}

\section{Models}

\begin{definition} [Algebra]
Let ${\cal L}$ be a first-order language.  An {\em ${\cal
L}$-algebra} is a family formed by a set $M$ and for each symbol
$f$ of ${\cal L}$ of arity $n$, a function $\hat{f}$ from 
$M^n$ to $M$. The denotation $\llbracket t \rrbracket_{\phi}$ of a term 
$t$ for an assignment $\phi$ 
mapping variables to elements of $M$
is defined as usual. 
\end{definition}

\begin{definition} [Model of a rewrite system]
Let ${\cal L}$ be a first-order language and $R$ an algorithm 
defined by a rewrite system on terms of the language ${\cal L}$. 
An {\em ${\cal L}$-algebra} ${\cal M}$ is a {\em model} of the
algorithm $R$, or the algorithm $R$ is 
{\em valid} in the model ${\cal M}$, (${\cal M} \models R$) if
for all rewrite rules 
$l \longrightarrow r$ of the rewrite system and valuations $\phi$, 
$\llbracket l \rrbracket_{\phi} =
\llbracket r \rrbracket_{\phi}$. 
\end{definition}

\begin{definition} [Model of an AC-rewrite system]
Let ${\cal L}$ be a first-order language containing binary 
function symbols $f_{1}, ..., f_{n}$, and $R$ an algorithm 
defined by an $AC(f_{1}, ..., f_{n})$-rewrite system on terms of the 
language ${\cal L}$. An {\em ${\cal L}$-algebra} ${\cal M}$ is a {\em model} 
of the algorithm $R$ (${\cal M} \models R$) if
\begin{itemize}
\item
for all rewrite rules 
$l \longrightarrow r$ of $R$ and valuations $\phi$, 
$\llbracket l \rrbracket_{\phi} =
\llbracket r \rrbracket_{\phi}$,
\item
for all valuations $\phi$ and indices $i$

$$\llbracket f_{i}(x, f_{i}(y,z)) \rrbracket_{\phi} = \llbracket 
f_{i}(f_{i}(x, y),  z) \rrbracket_{\phi}$$
$$\llbracket f_{i}(x, y) \rrbracket_{\phi} = \llbracket f_{i}(y,x) 
\rrbracket_{\phi}$$
\end{itemize}
\end{definition}

\begin{example}
Consider the language ${\cal L}$ formed by two binary symbols $+$ and
$\times$ and the algorithm $R$ defined by the rules

$$(x + y) \times z \longrightarrow (x \times z) + (y \times z)$$
$$x \times (y + z) \longrightarrow (x \times y) + (x \times z)$$
transforming for instance, the term $(a + a) \times a$ to 
the term $a \times a + a \times a$.
The structure $\langle \{0, 1\}, \mbox{\em min}, \mbox{\em
max} \rangle$
is a model of this algorithm. 
\end{example}

\begin{remark}
This definition of the validity of an algorithm in a model extends
some definitions of the semantics of a programming language where a
semantic is defined by a set $M$, a function $[~]$ mapping values of
the language to elements of $M$ and $n$-ary programs to functions from
$M^n$ to $M$, such that the program $P$ taking the values $v_{1}, ...,
v_{n}$ as input produces the value $w$ as output if and only if $[w] =
[P]([v_{1}], ..., [v_{n}])$.

Indeed, let us consider a programming language where the set of values is 
defined by a first-order language, whose symbols are called {\em
constructors}. Consider an extension of this language with a function 
symbol $p$ and possibly other function symbols. 
A program $P$ in this language is given by a terminating 
and confluent rewrite system on the extended language, such that for any 
$n$-uple of values $v_{1}, ..., v_{n}$
the program $P$ taking the values $v_{1}, ..., v_{n}$ as input produces 
the value $w$ as output if and only if 
the normal form of the term $p(v_{1}, ..., v_{n})$ is $w$.
Then, a model of this rewrite system is formed by
a set $M$,
for each constructor $c$ of arity $m$, a function $\hat{c}$ 
from $M^m$ to $M$,  a function $\hat{p}$ from $M^{n}$ to $M$, and possibly 
other functions, such that for all rules $l \longrightarrow r$ of the 
rewrite system and 
valuations $\phi$, 
$\llbracket l \rrbracket_{\phi} = \llbracket r \rrbracket_{\phi}$.

The denotations of the constructors define the function $[~]$ above
mapping values to 
elements of $M$ and the function $\hat{p}$ is the function $[P]$. 
For any $n$-uple of values $v_{1}, ..., v_{n}$, if the normal form 
of the term $p(v_{1}, ..., v_{n})$ is the value $w$ then 
$\llbracket w \rrbracket = \hat{p}(\llbracket v_{1} \rrbracket, ..., 
\llbracket v_{n} \rrbracket)$ and thus $[w] = [P]([v_{1}], ..., [v_{n}])$. 
\end{remark}

\section{Computing linear combinations of the unknowns}

\subsection{An algorithm}

Let ${\cal L}$ be a 2-sorted language with a sort $K$ for scalars and 
a sort $E$ for vectors containing 
two binary symbols $+$ and $\times$ of rank $\langle K, K, K \rangle$, 
two constants $0$ and $1$ of sort $K$,
a binary symbol, also written $+$, of rank
$\langle E, E, E \rangle$, a binary symbol $.$ of rank 
$\langle K, E, E \rangle$ and a constant ${\bf 0}$ of sort $E$. 

To transform a term of sort $E$ into a linear combination of the unknows, 
we want to develop sums of vectors 

$$\lambda . ({\bf u} + {\bf v}) \longrightarrow \lambda . {\bf u} + \lambda . {\bf v}$$
but factor sums of scalars and nested products

$$\lambda . {\bf u} + \mu . {\bf u} \longrightarrow (\lambda + \mu) . {\bf u}$$
$$\lambda . (\mu . {\bf u}) \longrightarrow (\lambda \times \mu) . {\bf u}$$
we also need the trivial rules

$${\bf u} + {\bf 0}  \longrightarrow {\bf u}$$
$$0 . {\bf u} \longrightarrow {\bf 0}$$
$$1 . {\bf u} \longrightarrow {\bf u}$$
and, finally, three more rules for confluence

$$\lambda . {\bf 0} \longrightarrow {\bf 0}$$
$$\lambda . {\bf u} + {\bf u} \longrightarrow (\lambda + 1) . {\bf u}$$
$${\bf u} + {\bf u} \longrightarrow (1 + 1) . {\bf u}$$

As we want to be able to apply the factorization rule to 
a term of the form $(3.{\bf x} + 4.{\bf y}) + 2.{\bf x}$, reductions in 
the above rewrite system must be defined modulo the associativity and 
commutativity of $+$. This leads to the following definition.

\begin{definition} [The rewrite system $R$]
The rewrite system $R$ is the AC(+)-rewrite system

$${\bf u} + {\bf 0}  \longrightarrow {\bf u}$$
$$0 . {\bf u} \longrightarrow {\bf 0}$$
$$1 . {\bf u} \longrightarrow {\bf u}$$
$$\lambda . {\bf 0} \longrightarrow {\bf 0}$$
$$\lambda . (\mu . {\bf u}) \longrightarrow (\lambda. \mu). {\bf u}$$
$$\lambda . {\bf u} + \mu . {\bf u} \longrightarrow (\lambda + \mu) . {\bf u}$$
$$\lambda . {\bf u} + {\bf u} \longrightarrow (\lambda + 1) . {\bf u}$$
$${\bf u} + {\bf u} \longrightarrow (1 + 1) . {\bf u}$$
$$\lambda . ({\bf u} + {\bf v}) \longrightarrow \lambda . {\bf u} + \lambda . {\bf v}$$
\end{definition}

\begin{definition} [Scalar rewrite system]
A {\em scalar rewrite system} is a 
rewrite system on a language containing at least the symbols $+$, 
$\times$, $0$ and $1$ such that:

\begin{itemize}
\item $S$ is terminating and ground confluent, 

\item for all closed terms
$\lambda$, $\mu$ and $\nu$, the pair of terms
\begin{itemize}
\item $0 + \lambda$ and $\lambda$,
\item $0 \times \lambda$ and $0$, 
\item $1 \times \lambda$ and $\lambda$, 
\item $\lambda \times (\mu + \nu)$ and $(\lambda \times \mu) + 
(\lambda \times \nu)$,
\item $(\lambda + \mu) + \nu$ and $\lambda + (\mu + \nu)$,
\item $\lambda + \mu$ and $\mu + \lambda$,
\item $(\lambda \times \mu) \times \nu$ and $\lambda \times (\mu \times \nu)$,
\item $\lambda \times \mu$ and $\mu \times \lambda$
\end{itemize}
have the same normal forms,

\item $0$ and $1$ are normal terms.
\end{itemize}
\end{definition}

We now want to prove that the for any scalar rewrite system $S$, the
system $R \cup S$ is terminating and confluent.

\subsection{Termination}

\begin{proposition}
\label{termR}
The system $R$ terminates.
\end{proposition}

\begin{proof}
Consider the following interpretation (compatible with AC) 

$$|{\bf u} + {\bf v}| = 2 + |{\bf u}| + |{\bf v}|$$
$$|\lambda . {\bf u}| = 1 + 2 |{\bf u}|$$
$$|{\bf 0}| = 0$$

Each time a term ${\bf t}$ rewrites to a term ${\bf t'}$ we have
$|{\bf t}| > |{\bf t'}|$. Hence, the system terminates.
\end{proof}

\begin{proposition}
\label{termRUS}
For any scalar rewrite system $S$, the system $R \cup S$ terminates.
\end{proposition}

\begin{proof}
By definition of the function $|~|$, if a term ${\bf t}$ $S$-reduces to a term
${\bf t'}$ then $|{\bf t}| = |{\bf t'}|$. 
Consider a $(R \cup S)$-reduction sequence. At each $R$-reduction step, the
measure of the term strictly decreases and at each $S$-reduction step
it remains the same. Thus there are only a finite number of $R$-reduction 
steps in the sequence and, as $S$ terminates, the sequence is finite. 
\end{proof}

\subsection{Confluence}

\begin{definition} [The rewrite system $S_0$]
The system $S_0$ is formed by the rules

$$0 + \lambda \longrightarrow \lambda$$
$$0 \times \lambda \longrightarrow 0$$
$$1 \times \lambda \longrightarrow \lambda$$
$$\lambda \times (\mu + \nu) \longrightarrow (\lambda \times \mu) + (\lambda \times \nu)$$
where $+$ and $\times$ are AC symbols.
\end{definition}

\begin{proposition}
\label{termS0}
The rewrite system $S_0$ terminates.
\end{proposition}

\begin{proof}
Consider the following interpretation (compatible with AC) 

$$||\lambda + \mu|| = ||\lambda|| + ||\mu|| + 1$$
$$||\lambda \times \mu|| = ||\lambda|| ||\mu||$$
$$||0|| = ||1|| = 2$$

Notice that all terms are worth at least $2$ and thus that 
each time a term $t$ rewrites to a term $t'$ we have
$||t|| > ||t'||$. Hence, the system terminates.
\end{proof} 

\begin{proposition}
\label{termRUS0}
The system $R \cup S_0$ terminates.
\end{proposition}

\begin{proof}
By definition of the function $|~|$, if a term ${\bf t}$
$S_0$-reduces to a term 
${\bf t'}$ then $|{\bf t}| = |{\bf t'}|$. 
Consider a $(R \cup S_0)$-reduction sequence. At each $R$-reduction step, the
measure of the term strictly decreases and at each $S_0$-reduction step,
it remains the same. Thus there are only a finite number of $R$-reduction 
steps in the sequence and, as $S_0$ terminates, by Proposition
\ref{termS0}, the sequence is finite. 
\end{proof}

\begin{proposition}
\label{confRUS0}
The rewrite system $R \cup S_0$ is confluent.
\end{proposition}

\begin{proof}
As the system terminates by Proposition \ref{termRUS0}, 
it is sufficient to prove the all critical 
pair close. This can be mechanically checked, for instance using the system  
CIME\footnote{\tt http://cime.lri.fr/}. 
\end{proof} 

\begin{definition}[Subsumption]
A terminating and confluent relation $S$ {\em subsumes} a relation
$S_0$ if whenever $t~S_0~u$, $t$ and $u$ have the same $S$-normal
form.
\end{definition}

\begin{definition}[Commutation]
The relation $R$ {\em commutes} with the relation $R'$, if 
whenever 
${\bf t}~R~{\bf u}_{1}$ and 
${\bf t}~R'~{\bf u}_{2}$, there exists a term 
${\bf w}$ such that ${\bf u}_{1}~R'~{\bf w}$ and
${\bf u}_{2}~R~{\bf w}$.
\end{definition}

\begin{proposition}
\label{comm}
Let $S$ be a scalar rewrite system, then $R$ commutes with the 
reflexive-transitive closure $S^*$ of $S$.
\end{proposition}

\begin{proof}
We check this for each rule of $R$, using the fact that 
in the left member of a rule, each subterms of sort scalar is 
either a variables or $0$ or $1$, which are normal forms.
\end{proof}

\begin{proposition}[Key Lemma]
\label{Lemma}
Let $R$, $S$ and $S_0$ be three relations defined on a set such that
$S$ is terminating and confluent, 
$R \cup S$ terminates, 
$R \cup S_0$ is confluent, 
$S$
subsumes $S_0$ ans the relation $R$ commutes with 
$S^{*}$.
Then, the relation $R \cup S$ is confluent.
\end{proposition}

\begin{proof}
We write ${\bf t}{\downarrow}$ for the $S$-normal form of ${\bf t}$.
We define the relation $S^{\downarrow}$ by ${\bf
t}~S^{\downarrow}~{\bf u}$ if ${\bf u}$ is the $S$-normal form of
${\bf t}$ and the relation $R;S^{\downarrow}$ by ${\bf
t}~(R;S^{\downarrow})~{\bf u}$ if there exists a term ${\bf v}$ such
that ${\bf t}~{R}~{\bf v}~S^{\downarrow}~{\bf u}$.

First notice that, if ${\bf t}~R~{\bf u}$ then 
${\bf t} {\downarrow}~(R;S^{\downarrow})~{\bf u} {\downarrow}$ using
the commutation of $R$ and $S^{*}$ and the unicity of $S$-normal forms.
Thus 
if ${\bf t}~(R \cup S)^*~{\bf u}$ then 
${\bf t} {\downarrow}~(R;S^{\downarrow})^*~{\bf u} {\downarrow}$
simulating each $R$-reduction step by a $(R;S^{\downarrow})$-reduction
step on normal forms.
In a similar way, 
if ${\bf t}~(R \cup S_0)^*~{\bf u}$ then 
${\bf t}{\downarrow}~(R;S^{\downarrow})^*~{\bf u}{\downarrow}$, 
simulating each $R$-reduction step by a $(R;S^{\downarrow})$-reduction
step on normal forms and using the subsumption of $S_0$ by $S$ for
$S_0$-steps.

We then check that $R;S^{\downarrow}$ is locally 
confluent.
If ${\bf t}~(R;S^{\downarrow})~{\bf v}_{1}$ 
and ${\bf t}~(R;S^{\downarrow})~{\bf v}_{2}$ then there exist terms 
${\bf u}_{1}$ and 
${\bf u}_{2}$ such that 
${\bf t}~R~{\bf u}_{1}~S^{\downarrow}~{\bf v}_{1}$
and ${\bf t}~R~{\bf u}_{2}~S^{\downarrow}~{\bf v}_{2}$.
Thus, by confluence, of $R \cup S_0$, there exists 
a term ${\bf w}$ such that ${\bf u}_{1}~(R \cup S_0)^*~{\bf w}$ and 
${\bf u}_{2}~(R \cup S_0)^*~{\bf w}$. 
Thus 
${\bf u}_{1} {\downarrow}~(R;S^{\downarrow})^*~{\bf w} {\downarrow}$ and
${\bf u}_{2} {\downarrow}~(R;S^{\downarrow})^*~{\bf w} {\downarrow}$
i.e.
${\bf v}_{1}~(R;S^{\downarrow})^*~{\bf w} {\downarrow}$ and
${\bf v}_{2}~(R;S^{\downarrow})^*~{\bf w} {\downarrow}$. 

As the relation $R;S^{\downarrow}$ is locally confluent and terminating, 
it is confluent.  

Finally, if we have ${\bf t}~(R \cup S)^*~{\bf u}_{1}$ and 
${\bf t}~(R \cup S)^*~{\bf u}_{2}$ then 
we have 
${\bf t} {\downarrow}~(R;S^{\downarrow})^*~{\bf u}_{1} {\downarrow}$ and 
${\bf t} {\downarrow}~(R;S^{\downarrow})^*~{\bf u}_{2} {\downarrow}$. 
Thus, there exists a term ${\bf w}$ such that 
${\bf u}_{1} {\downarrow}~(R;S^{\downarrow})^*~{\bf w}$ and 
and 
${\bf u}_{2} {\downarrow}~(R;S^{\downarrow})^*~{\bf w}$.
Thus ${\bf u}_{1}~(R \cup S)^*~{\bf w}$ and 
${\bf u}_{2}~(R \cup S)^*~{\bf w}$.
\end{proof} 

\begin{proposition}
\label{confRUS}
Let $S$ be a scalar rewrite system.
The rewrite system $R \cup S$ is confluent on terms containing 
variables of sort $E$ but no variables of sort $K$.
\end{proposition}

\begin{proof}
We use the Key Lemma on the set of semi-open terms, i.e. terms 
with variables of sort $E$ but no variables of sort $K$. As $S$ is 
ground confluent and 
terminating, it is confluent and terminating on semi-open terms, 
by Proposition \ref {termRUS}, the system $R \cup S$ terminates, 
by Proposition \ref{confRUS0}, the system $R \cup S_0$ is confluent, 
the system $S$ subsumes $S_0$ because $S$ is a scalar rewrite
system,
and by Proposition \ref{comm}, the system $R$ commutes with $S^*$.
\end{proof} 

\begin{remark}
Confluence on semi-open terms implies ground confluence in 
any extension of the language with constants for vectors, typically base
vectors.
\end{remark}

\subsection{Normal forms}

\begin{proposition}
\label{classification}
Let ${\bf t}$ be a normal term whose variables are among 
${\bf x}_{1}, ..., {\bf x}_{n}$. 
The term ${\bf t}$ is ${\bf 0}$ or a term of the form 
$\lambda_{1} . {\bf x}_{i_{1}} + ... +
\lambda_{k} . {\bf x}_{i_{k}} + 
{\bf x}_{i_{k+1}} + ... +
{\bf x}_{i_{k+l}}$
where the indices $i_{1}, ..., i_{k+l}$ are distinct
and $\lambda_{1}, ..., \lambda_{k}$ are neither $0$ nor $1$.

\end{proposition}

\begin{proof}
The term ${\bf t}$ is a sum ${\bf u}_{1} + ... + {\bf u}_{n}$ of
normal terms that are not sums (we take $n = 1$ if ${\bf t}$ is not a
sum). 

A normal term that is not a sum is either ${\bf 0}$, 
a variable, or a term of the form $\lambda . {\bf v}$. In this case, 
$\lambda$ is neither $0$ nor $1$ and 
${\bf v}$ is neither ${\bf 0}$, nor a sum of two vectors 
nor a product of a scalar by a vector, thus it is a variable.

As the term ${\bf t}$ is normal, if $n > 1$ then none of the ${\bf
u}_{i}$ is ${\bf 0}$. Hence, the term ${\bf t}$ is either ${\bf 0}$
or a term of the form

$$\lambda _{1} . {\bf x}_{i_{1}} + ... + \lambda _{k} . {\bf x}_{i_{k}} + 
{\bf x}_{i_{k+1}} + ... + {\bf x}_{i_{k+l}}$$
where $\lambda_{1}, ..., \lambda_{k}$ are neither $0$ nor $1$.
As the term {\bf t} is normal, the indices $i_{1}, ..., i_{k+l}$ are distinct.
\end{proof}

\section{Vectorial spaces}

Given a field 
${\cal K} = \langle K, +, \times, 0, 1 \rangle$
the class of ${\cal K}$-vectorial spaces can be defined as follows.

\begin{definition} [Vectorial space]
\label{def1} 
The structure $\langle E, +, .,
{\bf 0} \rangle$ is a ${\cal K}$-vectorial space if and only if the 
structure $\langle K, +, \times, 0, 1, 
E, +, ., {\bf 0} \rangle$ is a model of the 2-sorted theory.

$$\forall {\bf u} \forall {\bf v} \forall {\bf w}~(
({\bf u} + {\bf v}) + {\bf w} = {\bf u} + ({\bf v} + {\bf w}))$$
$$\forall {\bf u} \forall {\bf v}~(
{\bf u} + {\bf v}  = {\bf v} + {\bf u})$$
$$\forall {\bf u}~(
{\bf u} + {\bf 0}  = {\bf u})$$
$$\forall {\bf u}~\exists {\bf u'}~({\bf u} + {\bf u'} = {\bf 0})$$
$$\forall {\bf u}~
(1. {\bf u} = {\bf u})$$
$$\forall \lambda \forall \mu \forall {\bf u}~
(\lambda . (\mu. {\bf u}) = (\lambda . \mu). {\bf u})$$
$$\forall \lambda \forall \mu \forall {\bf u}~
((\lambda + \mu) . {\bf u} = \lambda . {\bf u} + \mu . {\bf u})$$
$$\forall \lambda \forall {\bf u} \forall {\bf v}~
(\lambda . ({\bf u} + {\bf v}) = \lambda . {\bf u} + \lambda . {\bf v})$$
\end{definition} 

We now prove that, the class of ${\cal K}$-vectorial spaces can be 
defined as the class of models of the rewrite system $R$.

\begin{proposition}
Let ${\cal K} = \langle K, +, \times, 0, 1 \rangle$
be a field. The structure $\langle E, +, .,
{\bf 0} \rangle$ is a ${\cal K}$-vectorial space if and only if the 
structure $\langle K, +, \times, 0, 1, 
E, +, ., {\bf 0} \rangle$ is a model of the 
rewrite system $R$.
\end{proposition}

\begin{proof}
We first check that all the rules of $R$ are valid in all vectorial spaces, 
i.e. that the propositions 

$$({\bf u} + {\bf v}) + {\bf w} = {\bf u} + ({\bf v} + {\bf w})$$
$${\bf u} + {\bf v} = {\bf v} + {\bf u}$$
$${\bf u} + {\bf 0}  = {\bf u}$$
$$0 . {\bf u} = {\bf 0}$$
$$1 . {\bf u} = {\bf u}$$
$$\lambda . {\bf 0} = {\bf 0}$$
$$\lambda . (\mu . {\bf u}) = (\lambda. \mu). {\bf u}$$
$$\lambda . {\bf u} + \mu . {\bf u} = (\lambda + \mu) . {\bf u}$$
$$\lambda . {\bf u} + {\bf u} = (\lambda + 1) . {\bf u}$$
$${\bf u} + {\bf u} =  (1 + 1) . {\bf u}$$
$$\lambda . ({\bf u} + {\bf v}) = \lambda . {\bf u} + \lambda . {\bf v}$$
are theorems of the theory of vectorial spaces. 

Seven of them are axioms of the theory of vectorial spaces, the propositions
$\lambda . {\bf u} + {\bf u} = (\lambda + 1) . {\bf u}$
and
${\bf u} + {\bf u} =  (1 + 1) . {\bf u}$
are consequence of $1 . {\bf u} = {\bf u}$ and 
$\lambda . {\bf u} + \mu . {\bf u} = (\lambda + \mu) . {\bf u}$.
Let us prove that $0. {\bf u} = {\bf 0}$.
Let ${\bf u'}$ be such that  ${\bf u} + {\bf u'}  = {\bf 0}$. Then 
$0 . {\bf u} = 0 . {\bf u} + {\bf 0} 
= 0 . {\bf u} + {\bf u} + {\bf u'}
= 0 . {\bf u} + 1 . {\bf u} + {\bf u'}
= 1 . {\bf u} + {\bf u'}
= {\bf u} + {\bf u'} = {\bf 0}$. 
Finally $\lambda . {\bf 0} = {\bf 0}$ is a consequence of 
$0. {\bf u} = {\bf 0}$ and 
$\lambda . (\mu . {\bf u}) = (\lambda. \mu). {\bf u}$.

Conversely, we prove that all axioms of vectorial spaces are 
valid in all models of $R$. The validity of each of them is a consequence
of the validity of a rewrite rule, except 
$\forall {\bf u} \exists {\bf u'}~({\bf u} + {\bf u'}  = {\bf 0})$
that is a consequence of 
${\bf u} + (-1) . {\bf u} = {\bf 0}$
itself being a consequence of 
$\lambda . {\bf u} + \mu . {\bf u} = (\lambda + \mu) . {\bf u}$
and $0 . {\bf u} = {\bf 0}$.
\end{proof}

\begin{proposition}[Universality]
Let ${\bf t}$ and ${\bf u}$ be two terms whose variables are among 
${\bf x}_{1}, ..., {\bf x}_{n}$. The following propositions are equivalent:
\begin{enumerate}
\item the normal forms of ${\bf t}$ and ${\bf u}$ are identical modulo AC,
\item the equation ${\bf t} = {\bf u}$ is valid in all ${\cal
K}$-vectorial spaces,
\item and the denotation of ${\bf t}$ and ${\bf u}$ in $K^n$ for the 
assignment $\phi = {\bf e}_{1} / {\bf x}_{1}, ..., {\bf e}_{n} / {\bf
  x_{n}}$,
where 
${\bf e}_{1}, ..., {\bf e}_{n}$ is the canonical base of $K^n$,
are identical.
\end{enumerate}
\end{proposition}

\begin{proof}
Proposition (i) implies proposition (ii) 
and proposition (ii) implies proposition (iii). Let us prove that
proposition (iii) implies proposition (i). 

Let ${\bf t}$ be a normal term 
whose variables are among ${\bf x}_{1}, ..., {\bf x}_{n}$. 
The {\em decomposition} of ${\bf t}$ along 
${\bf x}_{1}, ..., {\bf x}_{n}$ 
is the sequence
$\alpha_{1}, ..., \alpha_{n}$ 
such that if there is a subterm of the form 
$\lambda . {\bf x}_{i}$ in ${\bf t}$, 
then $\alpha_{i} = \lambda$,
if there is a subterm of the form 
${\bf x}_{i}$ in ${\bf t}$,
then $\alpha_{i} = 1$, and $\alpha_{i} = 0$ otherwise.

Assume $\llbracket {\bf t} \rrbracket_{\phi} = \llbracket {\bf u}
\rrbracket_{\phi}$. 
Let 
${\bf e}_{1}, ..., {\bf e}_{n}$ be the canonical base of 
$K^n$ and $\phi = {\bf e}_{1}/{\bf x}_{1}, ..., {\bf e}_{n}/{\bf x}_{n}$.
Call $\alpha_{1}, ..., \alpha_{n}$ the coordinates of 
$\llbracket {\bf t} \rrbracket_{\phi}$ 
in ${\bf e}_{1}, ..., {\bf e}_{n}$. 
Then the decompositions of the normal forms of ${\bf t}$ and ${\bf u}$ are 
both $\alpha_{1}, ..., \alpha_{n}$ and thus they are
identical modulo AC.
\end{proof}

\section{Bilinearity}

\subsection{An algorithm}

\begin{definition} [The rewrite system $R'$]
\label{R'}
Consider a language with four sorts: $K$ for scalars and $E$, $F$,
and $G$ for the vectors of three vector spaces, the symbols $+$, $\times$, 
$0$, $1$ for scalars, three copies of the symbols $+$, $.$ and ${\bf 0}$ for 
each sort $E$, $F$, and $G$ and a symbol $\otimes$ of rank 
$\langle E, F, G \rangle$. 

The system $R'$ is the rewrite system formed by three copies of 
the rules of the system $R$ and the rules

$$({\bf u} + {\bf v}) \otimes {\bf w} \longrightarrow
({\bf u} \otimes {\bf w}) + ({\bf v} \otimes {\bf w})$$
$$(\lambda . {\bf u}) \otimes {\bf v} \longrightarrow
\lambda . ({\bf u} \otimes {\bf v})$$
$${\bf u} \otimes ({\bf v} + {\bf w}) \longrightarrow
({\bf u} \otimes {\bf v}) + ({\bf u} \otimes {\bf w})$$
$${\bf u} \otimes (\lambda . {\bf v}) \longrightarrow 
\lambda . ({\bf u} \otimes {\bf v})$$
$${\bf 0} \otimes {\bf u} \longrightarrow {\bf 0}$$
$${\bf u} \otimes {\bf 0} \longrightarrow {\bf 0}$$
\end{definition}

\begin{proposition}
The rewrite system $R'$ terminates.
\end{proposition}

\begin{proof}
We extend the interpretation of Definition \ref{termR} with 

$$|{\bf u} \otimes {\bf v}| = (3 |{\bf u}| + 2)(3 |{\bf v}| + 2)$$
\end{proof}

\begin{proposition}
For any scalar rewrite system $S$, the system $R' \cup S$ terminates.
\end{proposition}

\begin{proof}
As in Proposition \ref{termRUS}.
\end{proof}

\begin{proposition}
\label{termR'US0}
The system $R' \cup S_0$ terminates.
\end{proposition}

\begin{proof}
As in Proposition \ref{termRUS0}.
\end{proof}

\begin{proposition}
The rewrite system $R' \cup S_0$ is confluent.
\end{proposition}

\begin{proof}
As in the proof of Proposition \ref{confRUS0}, we 
prove local confluence by checking that all critical pair close.
\end{proof}

\begin{proposition}
Let $S$ be a scalar rewrite system, then $R'$ commutes with 
$S^*$.
\end{proposition}

\begin{proof}
As in the proof of Proposition \ref{comm}. 
\end{proof}

\begin{proposition}
Let $S$ be a scalar rewrite system.
The rewrite system $R' \cup S$ is confluent on terms containing
variables of sort $E$, $F$, and $G$ but no variables of sort $K$.
\end{proposition}

\begin{proof}
Using the Key Lemma.
\end{proof}

\begin{proposition}
Let ${\bf t}$ be a normal term whose variables of sort $E$ are among 
${\bf x}_{1}, ..., {\bf x}_{n}$, whose variables of sort $F$ are 
among ${\bf y}_{1}, ..., {\bf y}_{p}$, and that has 
no variables of sort $G$ and $K$. 
If ${\bf t}$ has sort 
$E$ or $F$, then it has the same form as in Proposition 
\ref{classification}. If it has sort $G$, then 
it has the form 

$$\lambda_{1} . ({\bf x}_{i_{1}} \otimes {\bf y}_{j_{1}})
+ ... + 
\lambda_{k} . ({\bf x}_{i_{k}} \otimes {\bf y}_{j_{k}})
+ 
({\bf x}_{i_{k+1}} \otimes {\bf y}_{j_{k+1}})
+ ... +
({\bf x}_{i_{k+l}} \otimes {\bf y}_{j_{k+l}})$$
where the pairs of indices $\langle i_{1},j_{1} \rangle, ..., 
\langle i_{k+l}, j_{k+l} \rangle$ are distinct
and $\lambda_{1}, ..., \lambda_{k}$ are neither $0$ nor $1$.
\end{proposition}

\begin{proof}
The term ${\bf t}$ is a sum ${\bf u}_{1} + ... + {\bf u}_{n}$ of
normal terms that are not sums (we take $n = 1$ if ${\bf t}$ is not a
sum). 

A normal term that is not a sum is either ${\bf 0}$, 
a term of the form ${\bf v} \otimes {\bf w}$, 
or of the form $\lambda . {\bf v}$. In this case, 
$\lambda$ is neither $0$ nor $1$ and 
${\bf v}$ is neither ${\bf 0}$, nor a sum of two vectors 
nor a product of a scalar by a vector, thus it is of the form 
${\bf v} \otimes {\bf w}$. 

In a term of the form ${\bf v} \otimes {\bf w}$, neither 
${\bf v}$ nor ${\bf w}$ is a sum, a product of a scalar by a vector 
or ${\bf 0}$. Thus both ${\bf v}$ and ${\bf w}$ are variables. 

As the term ${\bf t}$ is normal, if $n > 1$ then none of the ${\bf
u}_{i}$ is ${\bf 0}$. Hence, the term ${\bf t}$ is either ${\bf 0}$
or a term of the form
$\lambda_{1} . ({\bf x}_{i_{1}} \otimes {\bf y}_{j_{1}})
+ ... + 
\lambda_{k} . ({\bf x}_{i_{k}} \otimes {\bf y}_{j_{k}})
+ 
({\bf x}_{i_{k+1}} \otimes {\bf y}_{j_{k+1}})
+ ... +
({\bf x}_{i_{k+l}} \otimes {\bf y}_{j_{k+l}})$
where $\lambda_{1}, ..., \lambda_{k}$ are neither $0$ nor $1$.
As the term {\bf t} is normal, the pairs of indices are distinct.
\end{proof}

\subsection{Bilinearity}

\begin{definition} [Bilinear operation]
\label{def1'}
Let $E$, $F$, and $G$ be three vectorial spaces on the same 
field. An operation $\otimes$ from $E \times F$ to $G$ is said to be 
{\em bilinear} if 

$$({\bf u} + {\bf v}) \otimes {\bf w} =
({\bf u} \otimes {\bf w}) + ({\bf v} \otimes {\bf w})$$
$$(\lambda . {\bf u}) \otimes {\bf v} =
\lambda . ({\bf u} \otimes {\bf v})$$
$${\bf u} \otimes ({\bf v} + {\bf w}) =
({\bf u} \otimes {\bf v}) + ({\bf u} \otimes {\bf w})$$
$${\bf u} \otimes (\lambda . {\bf v}) =
\lambda . ({\bf u} \otimes {\bf v})$$
\end{definition}

\begin{proposition}
Let ${\cal K} = \langle K, +, \times, 0, 1 \rangle$ be a field.
The structures 
$\langle E, +, ., {\bf 0} \rangle$, 
$\langle F, +, ., {\bf 0} \rangle$, 
$\langle G, +, ., {\bf 0} \rangle$ are ${\cal K}$-vectorial spaces and 
$\otimes$ is a bilinear operation from $E \times F$ to $G$ if and 
only if 
$\langle K, +, \times, 0, 1, 
E, +, ., {\bf 0},
F, +, ., {\bf 0},
G, +, ., {\bf 0}, \otimes \rangle$ is a model of the system $R'$. 
\end{proposition}

\begin{proof}
The validity of the rules of the three copies of the system $R$, express
that 
$\langle E, +, ., {\bf 0} \rangle$, 
$\langle F, +, ., {\bf 0} \rangle$, 
$\langle G, +, ., {\bf 0} \rangle$ are ${\cal K}$-vectorial spaces.
The validity of the six other rules is the validity of the axioms of
Definition \ref{def1'} plus the two extra propositions
${\bf 0} \otimes {\bf u} = {\bf 0}$ and
${\bf u} \otimes {\bf 0} = {\bf 0}$
that are consequences of these axioms.
\end{proof}

\begin{definition} [Tensorial product]
Let $E$ and $F$ be two vectorial spaces, the pair formed by the vectorial 
space $G$ and the bilinear operation from $E \times F$ to $G$ is a
{\em tensorial product} of $E$ and $F$ if for all bases 
$({\bf e}_{i})_{i \in I}$ of $E$ and $({\bf e'}_{j})_{j \in J}$ of $F$ the 
family $({\bf e}_{i} \otimes {\bf e'}_{j})_{\langle i,j \rangle}$ is a 
base of $G$. 
\end{definition}

\begin{example}
Let $\otimes$ be the unique bilinear operation such 
that 
${\bf e}_{i} \otimes {\bf e'}_{j} = {\bf e''}_{p(i-1)+j}$
where 
${\bf e}_{1}, ..., {\bf e}_{n}$ is the canonical base of $K^n$, 
${\bf e'}_{1}, ..., {\bf e'}_{p}$ that of of $K^p$, and 
${\bf e''}_{1}, ..., {\bf e''}_{np}$ that of $K^{np}$. 
Then $K^{np}$ together with $\otimes$ is the tensorial product of 
$K^n$ and $K^p$. 
\end{example}

\begin{proposition}[Universality]
Let ${\bf t}$ and ${\bf u}$ be two terms whose variables of sort $E$ are among 
${\bf x}_{1}, ..., {\bf x}_{n}$, whose variables of sort $F$ are 
among ${\bf y}_{1}, ..., {\bf y}_{p}$, and that have 
no variables of sort $G$ and $K$. The following propositions are equivalent:
\begin{enumerate}
\item the normal forms of ${\bf t}$ and ${\bf u}$ are identical modulo AC,
\item the equation ${\bf t} = {\bf u}$ is valid in all structures formed 
by three vectorial spaces and a bilinear operation,
\item the equation ${\bf t} = {\bf u}$ is valid in all structures formed 
by two vectorial spaces and their tensorial product,
\item and the denotation of ${\bf t}$ and ${\bf u}$ in $K^{np}$ for the 
assignment 

$$\phi = {\bf e}_{1} / {\bf x}_{1}, ..., {\bf e}_{n} / {\bf x_{n}}, 
{\bf e'}_{1} / {\bf y}_{1}, ..., {\bf e'}_{p} / {\bf y_{p}}$$
where 
${\bf e}_{1}, ..., {\bf e}_{n}$ is the canonical base of $K^n$,
${\bf e'}_{1}, ..., {\bf e'}_{p}$ that of $K^p$
and $\otimes$ is the 
unique bilinear operation such 
that 
${\bf e}_{i} \otimes {\bf e'}_{j} = {\bf e''}_{p(i-1)+j}$
where 
${\bf e''}_{1}, ..., {\bf e''}_{np}$ is the canonical base of $K^{np}$. 
\end{enumerate}
\end{proposition}

\begin{proof}
Proposition (i) implies proposition (ii), 
proposition (ii) implies proposition (iii) and 
proposition (iii) implies proposition (iv).
Let us prove that proposition (iv) implies proposition (i). 

Let ${\bf t}$ be a normal term of sort $G$ with variables 
of sort $E$ among ${\bf x}_{1}, ..., {\bf x}_{n}$, variables 
of sort $F$ among
${\bf y}_{1}, ..., {\bf y}_{p}$, and no variables of sort $G$ and $K$. 
The {\em decomposition} of ${\bf t}$ along 
${\bf x}_{1}, ..., {\bf x}_{n}$, 
${\bf y}_{1}, ..., {\bf y}_{p}$, 
is the sequence
$\alpha_{1}, ..., \alpha_{np}$ 
such that if there is a subterm of the form 
$\lambda . ({\bf x}_{i} \otimes {\bf y}_{j})$ in ${\bf t}$, 
then $\alpha_{p (i - 1) +j} = \lambda$,
if there is a subterm of the form 
${\bf x}_{i} \otimes {\bf y}_{j}$ in ${\bf t}$,
then $\alpha_{p (i - 1) + j} = 1$, and $\alpha_{p (i - 1) + j} = 0$ otherwise.

Assume $\llbracket {\bf t} \rrbracket_{\phi} = \llbracket {\bf u}
\rrbracket_{\phi}$. 
Call $\alpha_{1}, ..., \alpha_{np}$ the coordinates of 
$\llbracket {\bf t} \rrbracket_{\phi}$ 
in ${\bf e''}_{1}, ..., {\bf e''}_{np}$. 
Then the decompositions of the normal forms of ${\bf t}$ and ${\bf u}$ are 
both $\alpha_{1}, ..., \alpha_{np}$ and thus they are
identical modulo AC.
\end{proof}

\section*{Conclusion}

We usually define an algebraic structure by three components: a set,
some operations defined on this set and some propositions that must be
valid in the structure.
For instance a ${\cal K}$-vectorial space is defined
by a set $E$, the operations ${\bf 0}$, $+$ and $.$ and the equations
of Definition \ref{def1}. 

We can, in a more computation-oriented way, define an algebraic
structure by a set, operations on this set and an algorithm on terms
constructed upon 
these operations that must be valid in the structure.
For instance a ${\cal K}$-vectorial space is defined
by a set $E$, the operations ${\bf 0}$, $+$ and $.$ and the 
algorithm $R$.

This algorithm is a well-known algorithm in linear algebra: it is the
algorithm that transforms any linear expression into a linear
combination of the unknowns. This algorithm is, at a first look,
only one among the many algorithms used in linear algebra, 
but it completely
defines the notion of vectorial space: a vectorial space is any
structure where this algorithm is valid, it is any structure
where linear expressions can be transformed this way into linear
combinations of the unknowns.

\section*{Acknowledgements}

The authors want to thank \'Evelyne Contejean, Claude Kirchner and
Claude March\'e for comments on a previous draft of this paper.

\end{document}